\documentclass[11pt]{amsart}

\usepackage{amsmath}
\usepackage{amsthm}
\usepackage{amsfonts}
\usepackage{graphicx}
\usepackage{psfrag}
\usepackage{multirow}
\usepackage{enumerate}

\newcommand{\del}{{\partial}}

\theoremstyle{plain}
\newtheorem{prop}{{Proposition}}
\newtheorem{cor}{{Corollary}}
\newtheorem{lem}{{Lemma}}

\theoremstyle{definition}
\newtheorem*{defn}{{Definition}}

\begin{document}
\title[On blow-up sol'ns of the Jang equation in spherical symmetry]{On blow-up solutions of the Jang equation in spherical symmetry}
\author{Catherine Williams} 
\address{Department of Mathematics, Stanford University, Stanford, CA 94305}
\email{cathwill@math.stanford.edu}      
\date{\today}          

\begin{abstract}
We prove some related results concerning blow-up solutions for the Jang equation.
First: it has been shown that, given an outermost marginally outer trapped surface (MOTS) $\Sigma$, there exists a solution to Jang's equation which blows up at $\Sigma$. 
Here we show that in addition, large classes of spherically symmetric initial data have solutions to the Jang equation which blow up at non-outermost MOTSs, i.e.\ MOTSs which lie strictly inside of other MOTS, and even inside of strictly outer trapped surfaces. 
Unlike for outermost MOTSs, however, we show that there do not \textit{always} exist blow-up solutions for inner MOTSs, even in spherical symmetry.
Secondly, an unpublished result of R.\ Schoen, whose proof we include here, says that in the time-symmetric case, any MOTS  corresponding to a  blow-up solution for Jang's equation must be outer-area-minimizing, i.e.\ cannot be contained in a surface of strictly smaller area.
The statement is false without the assumption of time-symmetry, however; we construct an explicit spherically symmetric data set providing a counterexample for the general case.
\end{abstract}

\maketitle

\section{Introduction}

The aim of this paper is to address some questions concerning the relationship between non-outermost marginally outer trapped surfaces and blow-up solutions for Jang's equation.
These questions were largely motivated by an interest in the behavior of marginally outer trapped tubes and in hopes of finding a new tool with which to study them.

Our starting point is an initial data set $(\mathcal{M}, g, K)$ for the Einstein equations, where $\mathcal{M}$ is a Riemannian 3-manifold with metric $g$ and $K$ is a symmetric covariant 2-tensor on $\mathcal{M}$.
Alternately but equivalently, we could begin with a given spacetime $\mathcal{N}^4$ and consider a spacelike slice $\mathcal{M}^3$ with induced metric $g$ and second fundamental form $K$.
We do not make use of the Einstein equations \textit{per se} except via the dominant energy condition.
The initial data sets of interest here are those which are asymptotically flat ---
that is, we shall assume that, outside of a compact set, $\mathcal{M}$ is diffeomorphic to a finite number of copies of $\mathbb{R}^3$ minus a ball, and furthermore that with respect to the pulled-back Cartesian coordinates from $\mathbb{R}^3$, $g$ tends to the (flat) Euclidean metric and $K$ tends to $0$ as $|x| \rightarrow \infty$ on each copy.
(For this definition to be truly meaningful, one should specify rates of decay for the components of $g$ and $K$, and we address this in the next section.)
Without any loss of generality, we restrict our attention henceforth to manifolds with a single asymptotically  flat end.
Given such a manifold, we consider vectors to be outward-directed if they point towards this asymptotically flat end.

For a given spacelike surface $\Sigma \subset \mathcal{M}$, the \textit{inner} and \textit{outer expansions} on $\Sigma$, denoted by $\theta^\pm$,  are defined by
\[ \theta^\pm = \text{tr}_g K - K(\nu, \nu) \pm H, \]
where $\nu$ is the outward-directed unit normal along $\Sigma$ and $H$ is the mean curvature of $\Sigma$ in $\mathcal{M}$ with respect to $\nu$.
A \textit{marginally outer trapped surface} (MOTS) is a closed, spacelike surface $\Sigma \subset \mathcal{M}$ whose outer expansion $\theta^+$ vanishes at every point.
Along with their cousins \textit{marginally inner trapped surfaces} (MITS), whose inner expansion $\theta^-$ vanishes at every point, such surfaces are sometimes called \textit{apparent horizons}.

The Jang equation is a quasi-linear elliptic partial differential equation for a function $f$ on $\mathcal{M}$.
The most geometric way of expressing it, as in  \cite{SY81}, is to consider the graph of a function $f \in C^2(\mathcal{M})$ as a hypersurface of the product manifold $(\mathcal{M} \times \mathbb{R}, g + dt^2)$. 
The Jang equation is then
\begin{equation} \mathcal{J}[f] = \mathcal{H}[f] - \mathcal{P}[f] = 0, \label{Jang1} \end{equation}
where $\mathcal{H}[f]$ denotes the mean curvature of $\text{graph}f \subset \mathcal{M}\times\mathbb{R}$ computed with respect to its downward-pointing unit normal, and $\mathcal{P}[f] = \text{tr}_{\text{graph}f} \overline{K}$, where $\text{tr}_{\text{graph}f}$ denotes the trace with respect to the induced metric on ${\text{graph}f}$ and $\overline{K}$ is the symmetric 2-tensor on $\mathcal{M}\times \mathbb{R}$ obtained by extending $K$ trivially along the $\mathbb{R}$-factor.

In \cite{SY81}, Schoen and Yau showed that the existence of one or more apparent horizons in an asymptotically flat initial data set $(\mathcal{M}, g, K)$ is the only obstruction to the existence of an asymptotically decaying solution to \eqref{Jang1} on $\mathcal{M}$.
In particular, given an apparent horizon $\Sigma \subset \mathcal{M}$, a solution to \eqref{Jang1} on the region exterior to $\Sigma$ may blow up (or down) at $\Sigma$, so that the corresponding component of $\text{graph}{f} \subset \mathcal{M} \times \mathbb{R}$ has an asymptotically cylindrical end which approaches $\Sigma \times \mathbb{R}$.
Motivated by this result, others have since used the Jang equation to prove the existence of apparent horizons in various settings, notably in  \cite{SY83}, \cite{Y01}, \cite{E} and \cite{AM09}.
Furthermore, in \cite{Metz}, Metzger shows that if $\Sigma$ is an outermost MOTS, then in fact there \textit{must} exist a solution to \eqref{Jang1} which blows up at $\Sigma$ (and only at $\Sigma$, provided it also does not lie inside of any MITSs).

Such results suggest that finding blow-up solutions for the Jang equation could be a useful technique for locating MOTSs in initial data sets and even tracking their evolution over time.
For instance, by identifying the corresponding blow-up solution on each slice $\mathcal{M}_t = \{ t \} \times \mathcal{M}$ of a spacetime $\mathcal{N} = (-T, T) \times \mathcal{M}$, one could follow the motion of the outermost MOTSs with respect to $t$.
Indeed, a useful notion in this context is that of a \textit{marginally outer trapped tube} (MOTT), a hypersurface $\mathcal{H}$ of a spacetime which is foliated by MOTSs.
(MOTTs are generalizations of \textit{marginally trapped tubes} (MTTs), which are called \textit{dynamical} or \textit{isolated horizons} if they are spacelike or null, respectively, and are sometimes considered to be reasonable models of black hole boundaries; see, e.g., \cite{A}.)
The definition of a MOTT makes sense without reference to a background foliation of the spacetime by spacelike slices, but given such a slicing, a MOTT $\mathcal{H}$ is said to be adapted to it if each of its foliating MOTSs is contained in the intersection of $\mathcal{H}$ with one of the spacelike slices.
One might then hope to locate whole (adapted) MOTTs by tracking blow-up solutions to Jang's equation on successive slices of the spacetime.

A complication with such a program arises, however, in that the intersection of a MOTT $\mathcal{H}$ with a given slice $\mathcal{M}_t$ may have both outermost and non-outermost (inner) components. 
This is indeed the case if the MOTT ``weaves through time," as schematically depicted in Figure \ref{jumps}a, or if two separate MOTTs ``coalesce" into one, as indicated in Figure \ref{jumps}b; in both cases, the outermost MOTS jumps abruptly outward at time $t_1$. 
See \cite{AMMS} for a detailed analysis of such jumps.

%%%%%%%%%%%%%%%%%%%%%%%%%%%%%%%%%%%%%%%%%%%%%
%
%		figure
%
%%%%%%%%%%%%%%%%%%%%%%%%%%%%%%%%%%%%%%%%%%%%%
%
\begin{figure}[hbtp]
\begin{center}
{
\psfrag{t0}{\large{$t_0$}}
\psfrag{t1}{\large{$t_1$}}
\psfrag{t2}{\large{$t_2$}}
\psfrag{H}{\large{}}
\psfrag{infty}{\large{toward asymptotically flat end}}
\psfrag{sig1}{\large{}}
\psfrag{sig2}{\large{}}
\psfrag{sigtw}{\large{}}
\psfrag{???}{\large{ ?}}
\psfrag{(a)}{\large{(a)}}
\psfrag{(b)}{\large{(b)}}
\resizebox{5in}{!}{\includegraphics{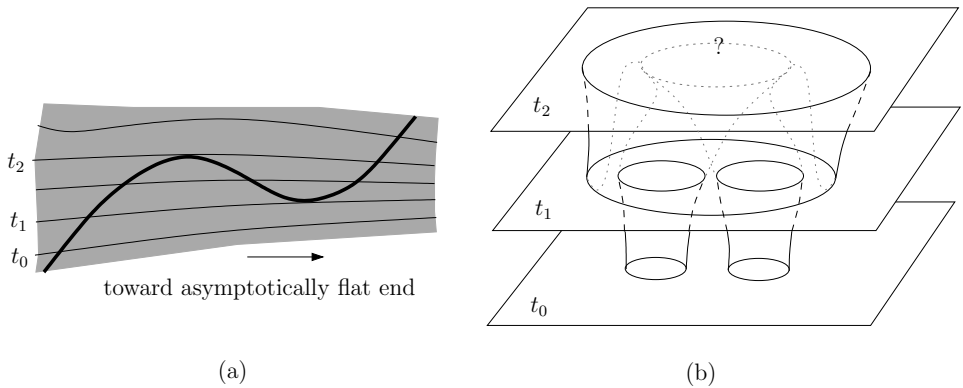}}
}
\vspace*{-.2in}
\caption{(a) 
A MOTT (thick line) in a spacetime foliated by spacelike hypersurfaces (thin lines). 
Although the MOTT is itself smooth, the outermost MOTS jumps outward at time $t_1$, and for $t_1 \leq t < t_2$, there are both outer and inner MOTSs on each slice.  
(b) 
Two MOTTs coalescing into one, again in a spacetime foliated by spacelike hypersurfaces.
For $t < t_1$, the outermost MOTS has two connected components, but at $t_1$ a new outermost MOTS appears, enclosing the 
previous two.  
Whether the two inner MOTSs persist for $t > t_1$ and connect up smoothly with the outer one to form a single immersed MOTT, as suggested by the dotted gray lines, is unknown.}
\end{center}\label{jumps}
\end{figure}
%
%%%%%%%%%%%%%%%%%%%%%%%%%%%%%%%%%%%%%%%%%%%%%
%
%		end figure
%
%%%%%%%%%%%%%%%%%%%%%%%%%%%%%%%%%%%%%%%%%%%%% 

In Section \ref{setup}, we cast the setup from \cite{SY81} into a spherically symmetric setting, in which the Jang equation becomes an ODE and hence much simpler to analyze than in general, and in Section \ref{ex}, we complete the picture of how Jang equation blow-up solutions correlate with the type of MOTT behavior depicted in Figure 1a (in spherical  symmetry).
In particular, at $t_0$, and indeed for $t < t_1$, there exists a single blow-up solution on each slice, corresponding to the sole MOTS (Proposition 1).  
At time $t_1$, a second blow-up solution appears corresponding to the new outer MOTS (Proposition 1 again), while the now inner MOTS continues to have its own separate blow-up solution (Proposition 2); outside of the outer MOTS where both solutions are well-defined, the former is larger than the latter.
At least for a short time after $t_1$, the inner blow-up solution continues to exist (Corollary 1), but it does not necessarily persist until time $t_2$ (Corollary 2).
% REFER TO FIGURE?
In short, we arrive at the least satisfying picture: a blow-up solution may or may not identify a MOTS which is outermost, and an inner MOTS may or may not have a blow-up solution corresponding to it.
Such a negative result suggests that the inner portion of the MOTT depicted in Figure 1b is also unlikely to be accessible to study via the Jang equation. 

The results in Section \ref{area} concern a different possible strategy for characterizing those MOTSs which correspond to blow-up solutions, namely whether or not they minimize area among those surfaces containing them.
This approach is motivated by a result due to Rick Schoen (Proposition 3) which says that in the general time-symmetric case (not necessarily spherically symmetric), any MOTS corresponding to a blow-up solution for Jang's equation must be outer-area-minimizing --- that is, it cannot be contained in a surface of strictly smaller area.
We find that for non-time-symmetric data, however, that characterization need not hold (Proposition 4): we construct a spherically symmetric initial data set satisfying the dominant energy condition which contains a MOTS with a corresponding blow-up solution, even though another MOTS of strictly smaller area contains it.
So in general, blow-up solutions do not correspond to outer-area-minimizing MOTSs only.

%%%%%%%%%%%%%%%%%%%%%%%%%%%%%%%%%%%%%%%%%%%%%%%%%%%%%%%%%%%%%%%%%%%%%%%%%%%%%%%%%%%%%%%%%%%%%%%%%%%%%%%%%%%%%%%%%%%%%%%%%%%%%%%%%
%																%
%																%
%						SET-UP & DEFINITIONS								%
%																%
%																%
%%%%%%%%%%%%%%%%%%%%%%%%%%%%%%%%%%%%%%%%%%%%%%%%%%%%%%%%%%%%%%%%%%%%%%%%%%%%%%%%%%%%%%%%%%%%%%%%%%%%%%%%%%%%%%%%%%%%%%%%%%%%%%%%%

\section{The Jang equation in spherical symmetry}\label{setup}

Our initial set-up and notation largely follow that of \cite{MO'M}, except where noted.
Suppose we have a spherically symmetric spacelike slice $\mathcal{M}$ of a given spacetime.  
The spherical symmetry implies that the Riemannian metric $g$ on $\mathcal{M}$ may be written locally with respect to a radial coordinate $r$ as
\[ g = \varphi(r) dr^2 + R^2(r) ds^2, \]
and assuming that the spacetime is spherically symmetric near $\mathcal{M}$, the second fundamental form $K$ may be written as
\[ K = \varphi K_\ell(r) dr^2 + R^2 K_R (r) ds^2, \]
where $R(r)$ is the area-radius of the spherically symmetric 2-spheres, $\varphi$, $K_\ell$, and $K_R$ are smooth functions of $r$, the function $\varphi > 0$, and $ds^2$ is the usual round metric on 2-spheres.
It will afford us some convenience in what follows to replace the radial coordinate $r$ with a coordinate $\ell$ which measures unit proper distance in the radial direction.
That is, we set
\[ \ell(r) = \int^r \sqrt{\varphi(r_\ast)}dr_\ast  \quad \left(\longleftrightarrow \quad  \del_\ell = \frac{1}{\sqrt{\varphi}}\del_r \right).\]
Then with respect to the coordinate $\ell$, $g$ and $K$ may be written 
\[ g = d\ell^2 + R^2(\ell) ds^2 \]
and
\[ K = K_\ell(\ell) d\ell^2 + R^2 K_R (\ell) ds^2. \]
Note that
\[ \text{tr}K(\ell) = K_\ell(\ell) + 2K_R(\ell). \]
We further assume that $\mathcal{M}$ is asymptotically flat (with a single end). 
Thus $\ell$ tends to infinity along the asymptotically flat end --- 
indeed, without loss of generality we may set $\mathcal{M} = [\ell_0, \infty)$ ---
and $R(\ell) \rightarrow \ell$ and $K_\ell$, $K_R \rightarrow 0$ as $\ell \rightarrow \infty$; the exact rates are discussed below.

Ordinarily, the Jang equation is a second-order PDE for $f$, but
if we require $\text{graph}f$ to be spherically symmetric, then $f = f(\ell)$ and \eqref{Jang1} reduces to the first-order ODE\footnote{Note that, while our derivation of \eqref{Jang} closely follows that of Malec and {\'O} Murchadha in \cite{MO'M}, we are using the version of \eqref{Jang1} with a minus sign, as in \cite{SY81}, while they use a plus; the version of \eqref{Jang} appearing in \cite{MO'M} therefore has a plus sign in front of the $R^2$ term.}
\begin{equation} \del_\ell (R^2 k) - R^2 \left( \text{tr}K - K_\ell k^2 \right) = 0 
\label{Jang}
\end{equation}
where
\[ k = \frac{\del_\ell f}{\sqrt{1 + (\del_\ell f)^2}} \quad \left(\longleftrightarrow \quad \del_\ell f = \frac{k}{\sqrt{1-k^2}} \right). \]
Notice that only solutions with $-1 < k (\ell) < 1$  correspond to ``physical" (real-valued) solutions $f$, and that $|\del_\ell f|$ blows up precisely when $|k| = 1$.
One readily computes that the future inner and outer expansions of the round 2-sphere corresponding to $\ell$ are 
\begin{equation} \theta^{\pm}(\ell) = 2K_R(\ell) \pm H(\ell), \label{theta}\end{equation}
where
\begin{equation} H(\ell) = \frac{2(\del_\ell R)}{R}(\ell) \label{H} \end{equation}
 is the mean curvature of the round 2-sphere corresponding to $\ell$, so we may rewrite equation \eqref{Jang} in several other forms, each of which will be useful in a different context:
\begin{eqnarray}
\del_\ell k & = &  -K_\ell k^2 - H k + \text{tr}K \label{J1} \\
\del_\ell k & = &  K_\ell (1 - k^2) - H (1 + k) + \theta^+ \label{J2} \\
\del_\ell k & = &  K_\ell (1 - k^2) + H (1 - k) + \theta^-. \label{J3}
\end{eqnarray} 
Although we do not make use  of this fact, equation \eqref{J1} in particular makes it clear that this reduced version of Jang's equation is a (scalar) Riccati equation.

%%%%%%%%%%%%%%%%%%%%%%%%%%%%%%%%%%%%%%%%%%%%%%%%%%%%%%%%%%%%%%%%%%%%%%%%%%%%%%%%%%%%%%%%%%%%%%%%%%%%%%%%%%%%%%%%%%%%%%%%%%%%%%%%%
%																%
%																%
%					BLOW-UP SOL'NS FOR OUTERMOST MOTS							%
%																%
%																%
%%%%%%%%%%%%%%%%%%%%%%%%%%%%%%%%%%%%%%%%%%%%%%%%%%%%%%%%%%%%%%%%%%%%%%%%%%%%%%%%%%%%%%%%%%%%%%%%%%%%%%%%%%%%%%%%%%%%%%%%%%%%%%%%%

\section{Existence of blow-up solutions}\label{ex}

In spherically symmetry, a MOTS corresponds to a point $\ell$ at which $\theta^+(\ell) = 0$, and the following definitions are tailored to this setting. 

\begin{defn} Given a surface $\ell_\ast$ in $\mathcal{M}$, we say a solution $k(\ell) \in C^\infty([\ell_\ast, \infty))$ is a \textit{blow-up solution for} $\ell_\ast$ if $k(\ell_\ast) = -1$, $\del_\ell k (\ell_\ast) = 0$, 
$|k(\ell)| < 1$ for all $\ell > \ell_\ast$, and $k(\ell) \rightarrow 0$ as $\ell \rightarrow \infty$.
(If $\ell_\ast$  is a MOTS, then $\del_\ell k (\ell_\ast) = 0$ is automatically satisfied if $k(\ell_\ast) = -1$ by equation \eqref{J2}.)
\end{defn}
\begin{defn} A MOTS $\ell_\ast$ is \textit{locally outermost} if there exists an $\varepsilon > 0$ such that $\theta^+(\ell) > 0$ for all $\ell \in (\ell_\ast, \ell_\ast + \varepsilon)$, and it is \textit{(globally) outermost} if $\theta^+(\ell) > 0$ for all $\ell > \ell_\ast$.
We shall call a locally outermost MOTS $\ell_\ast$ \textit{finitely stable} if there exists $m \geq 1$ such that $\del^m_\ell \theta^+ (\ell_\ast) > 0$ --- that is, if $\theta^+$ has positive $m$-th variation at $\ell_\ast$ for some finite $m$.  
If $m = 1$, the MOTS is \textit{strictly stable} (cf. \cite{AMS}).
\end{defn}

In order to align with \cite{SY81}, we restrict our attention to initial data $(\mathcal{M} = [\ell_0, \infty), R, K_\ell, K_R)$ satisfying the decay conditions given there.
In particular, translating the rates given with respect to Cartesian coordinates in \cite{SY81} to our spherically symmetric setting, we shall require:
\begin{eqnarray}
\left| \frac{R^2}{\ell^2} -1 \right| & = & O(\ell^{-1}) \label{decay1}\\
\left| H - \frac{2}{\ell} \right| & = & O(\ell^{-2}) \label{decay2}\\
\left| \del_\ell H + \frac{2}{\ell^2} \right| & = & O(\ell^{-3}) \label{decay3}\\
\left| K_\ell \right|, \,\left| K_R \right| & = & O(\ell^{-2})  \label{decay4}\\
\left| \del_\ell K_{\ell} \right|, \,\left| \del_\ell K_{R} \right| & = & O(\ell^{-3}) \\
\left| \text{tr}K \right| & = & O(\ell^{-3}), \label{decay6}
\end{eqnarray}
where $H(\ell)$ is as in \eqref{H}. 
From the decay on  $H$, $K_\ell$, and $\text{tr}K$ in particular, we have the following useful lemma.
\begin{lem} If $(\mathcal{M}, R, K_\ell, K_R)$ is an initial data set satisfying decay conditions \eqref{decay1}-\eqref{decay6}, then any bounded solution $k\in C^1(\mathcal{M})$ to Jang's equation \eqref{J1} satisfies $k(\ell) \rightarrow 0$ as $\ell \rightarrow \infty$.  In particular, $k(\ell) = O(\ell^{-1})$. \label{lemma}
\end{lem}
\begin{proof}
Since $|K_\ell(\ell)| = O(\ell^{-2})$ and $|\text{tr}K(\ell)| =  O(\ell^{-3})$ (\eqref{decay4} and \eqref{decay6}, respectively), and we know that $|k(\ell)|$ is uniformly bounded on $\mathcal{M}$, equation \eqref{J1} yields
\[ \left| \del_\ell k + Hk \right|  = O(\ell^{-2}). \]
Since $|H - \frac{2}{\ell}| = O(\ell^{-2})$,
applying the triangle inequality implies
\[ \left| \del_\ell k + \frac{2}{\ell} \cdot  k \right| = O(\ell^{-2}), \]
and multiplying both sides by $\ell^2$ then yields
 \[ \left| \ell^2 (\del_\ell k) + 2\ell k \right| =  \left| \del_\ell(\ell^2 k ) \right|  \leq C \]
for some constant $C >0$.
Integrating both sides on an interval $[L, \ell]$, we have
\[ \left| \ell^2 k(\ell) - L^2 k(L) \right | \leq C(\ell - L)
\leq C\ell, \]
which in turn implies that
\[  \left| k(\ell) \right| \leq \frac{C}{\ell} + \frac{L^2 k(L)}{\ell^2} = O(\ell^{-1}). \]
\end{proof}

Henceforth we assume that our initial data $(\mathcal{M} = [\ell_0, \infty), R, K_\ell, K_R)$ satisfy the decay conditions \eqref{decay1}-\eqref{decay6} above.
We further require that the data satisfy \textit{either} $\theta^- < 0$ everywhere, \textit{or} $\text{tr}K  < 0$ everywhere.
With these assumptions in place, we have both old and new existence results.

\begin{prop} A blow-up solution exists for any outermost, finitely stable MOTS $\ell_\ast$.
\end{prop}

\noindent
\textit{Remarks.}
As previously mentioned, this result is not new;
in particular, Metzger proved a more general existence theorem from which one can derive this one in \cite{Metz}, based largely on techniques of Schoen-Yau and the barrier argument from Andersson-Metzger (\cite{SY81} and \cite{AM09}, respectively).
The spherically symmetric case is also explicitly addressed in \cite{MO'M}, and some portions of the proof below are similar to arguments made there.
We include a full proof here primarily so that we can build on its methods in the proof of Proposition 2 (which \textit{is} new).

\begin{proof}  
By standard ODE theory it is clear that there exists a local solution $k(\ell)$ to Jang's equation in a neighborhood of $\ell_\ast$ satisfying the initial value condition $k(\ell_\ast) = -1$.  
Since $\theta^+(\ell_\ast) = 0$, plugging this value into \eqref{J2} yields $\del_\ell k(\ell_\ast) = 0$.
We need to show that this solution in fact exists on all of $[\ell_\ast, \infty)$, $|k(\ell)| < 1$ for all $\ell > \ell_\ast$, and $k(\ell) \rightarrow 0$ as $\ell \rightarrow \infty$.

Differentiating \eqref{J2} yields
\begin{equation} \del^2_\ell k = (\del_\ell K_{\ell})(1-k^2) - (\del_\ell H)(1+k) + \del_\ell\theta^+ - 2( \del_\ell k) k K_\ell - (\del_\ell k)H, \label{full2ndderiv}
\end{equation}
and evaluating at $\ell_\ast$ tells us that
\begin{equation} \del^2_\ell k(\ell_\ast) = \del_\ell \theta^+(\ell_\ast). \label{2ndderiv} \end{equation} 
Since $\ell_\ast$ is outermost, it is stable, i.e.\ $\del_\ell \theta^+(\ell_\ast) \geq 0$.
If it is in fact strictly stable, $\del_\ast \theta^+(\ell_\ast) > 0$, then \eqref{2ndderiv} implies that 
$-1 < k(\ell)$ for $\ell \in (\ell_\ast, \ell_\ast + \varepsilon)$, some small $\varepsilon > 0$.
If $\ell_\ast$ is stable but not strictly so, then $\del_\ell \theta^+(\ell_\ast) = 0 $; then since $\ell_\ast$ is outermost, its second variation must be non-negative, $\del^2_{\ell}\theta^+(\ell_\ast) \geq 0$, and if it too is zero, we can consider the third variation, and so on.
Recursing and using our hypothesis that $\ell_\ast$ is finitely stable, we find that $\del^{j}_{\ell} \theta^+(\ell_\ast) > 0$ while $\theta^+(\ell_\ast) = \del_\ell \theta^+(\ell_\ast) = \del^2_\ell \theta^+(\ell_\ast) = \cdots = \del^{j-1}_\ell \theta^+(\ell_\ast) = 0$ for some $j \geq 2$.
By taking $j-1$ derivatives of \eqref{J2} and evaluating each at $\ell_\ast$, we find that $\del_\ell^{j+1}k (\ell_\ast) > 0$ while $\del_\ell k (\ell_\ast) = \cdots = \del_\ell^{j}k (\ell_\ast) = 0$.
Thus we may again conclude that $-1 < k(\ell)$ for $\ell \in (\ell_\ast, \ell_\ast + \varepsilon)$, some small $\varepsilon > 0$.

Suppose $\mathcal{D} \subset \mathbb{R}$ is the maximal domain of definition of $k$, and let $\mathcal{A} = \{ \ell \in (\ell_\ast, \infty) \cap \mathcal{D} : k(\ell) \leq -1 \}$.
If $\mathcal{A} \neq \emptyset$, set $L = \inf \mathcal{A}$.
Then $k(L) = -1$, and $L > \ell_\ast$ since $-1 < k(\ell)$ for $\ell \in (\ell_\ast, \ell_\ast + \varepsilon)$.  
Plugging $L$ into \eqref{J2}, we have $\del_\ell k(L) = \theta^+(L)$, and thus since $\ell_\ast$ is the outermost MOTS,  we must have $\theta^+(L) = \del_\ell k(L) > 0$.
But $\del_\ell k(L) > 0$ and $k(L) = -1$ together imply that $k(\ell) < -1$ for $\ell< L$  near $L$, contradicting the definition of $L$. 
So $\mathcal{A} = \emptyset$, and $k(\ell) > -1$ for all $\ell > \ell_\ast$ in $\mathcal{D}$.

For the upper bound on $k(\ell)$, we have two cases.
If the data satisfies $\theta^- < 0$ everywhere, we define $\mathcal{B} = \{ \ell \in (\ell_\ast, \infty) \cap \mathcal{D} : k(\ell) \geq 1 \}$, and set $L^\prime = \inf \mathcal{B}$ if $\mathcal{B} \neq \emptyset$. Then $k(L^\prime) = 1$, and since $k(\ell_\ast) = -1$, $L^\prime > \ell_\ast$.
Plugging $L^\prime$ into \eqref{J3} yields $\del_\ell k(L^\prime) = \theta^-(L^\prime) < 0$.
But $\del_\ell k(L^\prime) < 0$ and $k(L^\prime) = 1$ together imply that  $k(\ell) > 1$ for $\ell< L^\prime$  near $L^\prime$, contradicting the definition of $L^\prime$. 
Thus $\mathcal{B} = \emptyset$ and $k(\ell) < 1$ for all $\ell \geq \ell_\ast$ in $\mathcal{D}$.

If the data satisfies $\text{tr}K < 0$ everywhere, we instead define  $\mathcal{B} = \{ \ell \in (\ell_\ast, \infty) \cap \mathcal{D} : k(\ell) \geq 0 \}$, and set $L^\prime = \inf \mathcal{B}$ if $\mathcal{B} \neq \emptyset$.
Then since $L^\prime > \ell_\ast$, $k(L^\prime) = 0$, and $\del_\ell k(L^\prime) = \text{tr}K(L^\prime) < 0$, we derive a contradiction to the definition of $L^\prime$ exactly like the previous one, namely that $k(\ell) > 0$ for $\ell < L^\prime$ near $L^\prime$.
So $\mathcal{B} = \emptyset$, and this time we have $k(\ell) < 0$ for all $\ell \geq \ell_\ast$ in $\mathcal{D}$.

In either of the two cases, we may conclude that $|k(\ell)| < 1$ for all $\ell_\ast < \ell \in \mathcal{D}$, so since the functions $R, K_R$, and $K_\ell$ are smooth on $[\ell_\ast, \infty)$, we must have $[\ell_\ast, \infty) \subset \mathcal{D}$.
That is, $k$ is defined out to infinity and satisfies $|k(\ell)| < 1$ for all $\ell > \ell_\ast$.
Lemma \ref{lemma} then implies that $k(\ell) \rightarrow 0$ as $\ell \rightarrow \infty$, as desired.
\end{proof}

%%%%%%%%%%%%%%%%%%%%%%%%%%%%%%%%%%%%%%%%%%%%%%%%%%%%%%%%%%%%%%%%%%%%%%%%%%%%%%%%%%%%%%%%%%%%%%%%%%%%%%%%%%%%%%%%%%%%%%%%%%%%%%%%%
%																%
%																%
%					BLOW-UP SOL'NS FOR INNER MOTS								%
%																%
%																%
%%%%%%%%%%%%%%%%%%%%%%%%%%%%%%%%%%%%%%%%%%%%%%%%%%%%%%%%%%%%%%%%%%%%%%%%%%%%%%%%%%%%%%%%%%%%%%%%%%%%%%%%%%%%%%%%%%%%%%%%%%%%%%%%%

\begin{prop} Suppose $\ell_\ast$ is an outermost, finitely stable MOTS.  If $\ell_{\ast\ast} < \ell_\ast$ is another finitely stable MOTS which is outermost in $[\ell_{\ast\ast}, \ell_\ast)$ (that is, $\theta^+(\ell) > 0$ for $\ell_{\ast\ast} < \ell < \ell_\ast$),
then a blow-up solution exists for $\ell_{\ast\ast}$.   
\end{prop}
\begin{proof}
As in the proof of Proposition 1, it is immediate that a local solution $k(\ell)$ exists satisfying $k(\ell_{\ast\ast}) = -1$ and $\del_\ell k(\ell_{\ast\ast}) = 0$.
It again remains to show that this solution $k$ exists for all $\ell \in [\ell_{\ast\ast}, \infty)$, satisfies $|k(\ell)| < 1$ for all $\ell > \ell_{\ast\ast}$, and that $k(\ell) \rightarrow 0$ as $\ell \rightarrow \infty$.

Let $\mathcal{D} \subset \mathbb{R}$ be the maximal domain of definition of $k$, let $\mathcal{A} = \{ \ell \in (\ell_{\ast\ast}, \infty) \cap \mathcal{D} : k(\ell) \leq -1 \}$, and set $L = \inf \mathcal{A}$.
Since $\ell_{\ast\ast}$ is locally outermost and finitely stable, as in the proof of Proposition 1  we have $k(\ell) > -1$ for $\ell \in (\ell_{\ast\ast}, \ell_{\ast\ast} + \varepsilon)$ for some small $\varepsilon > 0$. 
Thus $L > \ell_{\ast\ast}$.
By definition of $L$, we must have $k(L) = -1$ and $\del_\ell k(L) \leq 0$.
But any point at which $k(\ell) = -1$ necessarily has $\del_\ell k(\ell) = \theta^+(\ell)$ (by equation \eqref{J2}), so since $\theta^+ > 0$ for all $\ell \in (\ell_{\ast\ast}, \infty) \setminus \{ \ell_\ast \}$, the only possibility is that $L = \ell_\ast$. 
So suppose $L = \ell_\ast$.
Since $\ell_\ast$ is a finitely stable MOTS, there exists $j \geq 1$ such that $\del^{j}_\ell \theta^+ (\ell_\ast) > 0$ while $\theta^+(\ell_\ast) = \del_\ell \theta^+(\ell_\ast) = \cdots = \del^{j-1}_\ell \theta^+ (\ell_\ast) = 0$, and since $\theta^+ > 0$ to either side of $\ell_\ast$, this $j$ must be even.
Computing each derivative of $k$ at $\ell_\ast$ recursively from \eqref{J2} and using the fact that $k(\ell_\ast) = k(L) = -1$, we then find that $\del^{j+1}_\ell k (\ell_\ast) > 0$ while $\del_\ell k (\ell_\ast) = \cdots = \del^{j}_\ell k (\ell_\ast) = 0$.
But since $j+1$ is odd, we must have $k(\ell) < -1$ for $\ell < \ell_\ast$ very close to $\ell_\ast$, contradicting the fact that $L = \ell_\ast$.
So in fact $\mathcal{A} = \emptyset$.

If the data satisfy $\theta^- < 0$ (resp.\ $\text{tr}K < 0$), we set $\mathcal{B} = \{ \ell \in (\ell_{\ast\ast}, \infty) \cap \mathcal{D} : k(\ell) \geq 1 \text{ (resp.\ 0)} \}$ and deduce that $\mathcal{B} = \emptyset$ exactly as in the proof of Proposition 1.
Thus $[\ell_{\ast\ast}, \infty) \subset \mathcal{D}$, i.e.\ $k(\ell)$ exists out to infinity, and $|k(\ell)| < 1$ on $(\ell_{\ast\ast}, \infty)$.
That $k(\ell) \rightarrow \infty$ as $\ell \rightarrow \infty$ now follows immediately from Lemma \ref{lemma}.
\end{proof}

An \textit{outer trapped surface} is a closed, spacelike surface for which $\theta^+ < 0$ at every point, so in our setting, it corresponds to a point $\ell$ at which $\theta^+(\ell) < 0$.
\begin{cor} There exist blow-up solutions for some finitely stable MOTSs which lie inside of (strictly) outer trapped surfaces.
\end{cor}
\begin{proof}
This follows essentially immediately from Proposition 2 by continuous dependence on parameters.
That is, suppose we smoothly perturb the initial data of Proposition 2 in a neighborhood of $\ell_\ast$, leaving it unchanged elsewhere, in such a way that $\theta^+(\ell_\ast \pm \varepsilon) = 0$ and  $\theta^+(\ell) < 0$ for $\ell_\ast - \varepsilon < \ell < \ell_\ast + \varepsilon$, some small $\varepsilon > 0$.
We may always choose $\varepsilon$ and the perturbation sufficiently small that the data retains the property that $\text{tr}K < 0$ or $\theta^- < 0$, whichever held for the original data.
Then the resulting solution $k_{\varepsilon}$ to the initial value problem \eqref{J1}, $k(\ell_{\ast\ast}) = -1$, will be close to the original solution $k_0$ for the unmodified data.
In particular, since the original solution $k_0 > -1$ near $\ell_\ast$, for sufficiently small $\varepsilon$ we have $k_\varepsilon > -1$ near $\ell_\ast$.
Hence $k_\varepsilon > -1$ for all $\ell > \ell_{\ast\ast}$ as well, since, as we saw in the proof of Proposition 2, the solution can only hit $-1$ where $\theta^+ \leq 0$.
Since either $\text{tr}K < 0$ or $\theta^- < 0$ everywhere, we again have an upper barrier for $k$ as in the proof of Proposition 2, from which it follows that $|k_\varepsilon| < 1$ for all $\ell > \ell_{\ast\ast}$.
Such solutions $k_\varepsilon$ are thus blow-up solutions for $\ell_{\ast\ast}$, which lies inside of trapped surfaces $\ell \in (\ell_\ast - \varepsilon, \ell_\ast + \varepsilon)$.
\end{proof} 

\begin{cor} There exist initial data sets containing finitely stable MOTSs interior to strictly outer trapped surfaces for which \textit{no} blow-up solution exists.
\end{cor}
\begin{proof}
Again we take as our starting point the initial data set of Proposition 2 and smoothly perturb the data in a neighborhood of $\ell_\ast$ so that $\theta^+(\ell_\ast \pm \varepsilon) = 0$ and  $\theta^+(\ell) < 0$ for $\ell_\ast - \varepsilon < \ell < \ell_\ast + \varepsilon$,  this time for some small \textit{fixed} $\varepsilon > 0$.
Next, we smoothly decrease $K_R$ within the interval $(\ell_\ast - \varepsilon,  \ell_\ast + \varepsilon)$, leaving $H$ and $K_\ell$ fixed, such that $\theta^+(\ell) < -\tau$ for $\ell_\ast - \frac 12 \varepsilon < \ell <  \ell_\ast + \frac 12 \varepsilon$, some $\tau \in (0, \infty)$.
(By choosing $\varepsilon$ and the initial perturbation sufficiently small at the outset, these modifications preserve whichever of $\text{tr}K < 0$ or $\theta^- < 0$ held initially, since decreasing $K_R$ while fixing $H$ and $K_\ell$ only decreases $\text{tr}K$ and $\theta^-$.)
Then by inspection of \eqref{J2}, it is clear that for sufficiently large $\tau$, the solution $k_\tau$ for the initial value problem \eqref{J1}, $k(\ell_{\ast\ast}) = -1$, will be forced to drop below $-1$ somewhere in the interval $(\ell_\ast - \frac 12 \varepsilon, \ell_\ast + \frac 12 \varepsilon)$.
Such a solution is therefore \textit{not} a blow-up solution for $\ell_{\ast\ast}$ in the corresponding initial data set, and by uniqueness of solutions to ODEs, there can be no other blow-up solution.
\end{proof} 
%\noindent
%\textit{Remark.}  Proposition 4 below shows that initial data sets as in Proposition 2 (and by continuity, Corollary 1) can be constructed to comply with the dominant energy condition, but it would be interesting to know whether the modifications outlined in Corollary 2 necessarily force the initial data to violate it. 

%%%%%%%%%%%%%%%%%%%%%%%%%%%%%%%%%%%%%%%%%%%%%%%%%%%%%%%%%%%%%%%%%%%%%%%%%%%%%%%%%%%%%%%%%%%%%%%%%%%%%%%%%%%%%%%%%%%%%%%%%%%%%%%%%
%																%
%																%
%			NON-EXISTENCE OF BLOW-UP SOL'NS FOR OUTER-AREA-MINIMIZING MOTS, K=0	  				%
%																%
%																%
%%%%%%%%%%%%%%%%%%%%%%%%%%%%%%%%%%%%%%%%%%%%%%%%%%%%%%%%%%%%%%%%%%%%%%%%%%%%%%%%%%%%%%%%%%%%%%%%%%%%%%%%%%%%%%%%%%%%%%%%%%%%%%%%%

\section{Blow-up solutions and area-minimization}\label{area}

The following is an unpublished result of Rick Schoen's.

\begin{prop}  In the time-symmetric case (not necessarily spherically symmetric), a MOTS can have a blow-up solution only if it is outer-area-minimizing.
\end{prop}
\noindent
\textit{Remark.}  In the non-spherically symmetric context, a blow-up solution for a MOTS $\Sigma \subset \mathcal{M}$ is a solution to \eqref{Jang1} on the exterior of $\Sigma$ whose graph in $\mathcal{M} \times \mathbb{R}$ asymptotically approaches the infinite cylinder $\Sigma \times \mathbb{R^+}$; cf.\ \cite{SY81} Corollary 2. 
\begin{proof}  Suppose otherwise.  
That is, suppose $(\mathcal{M}, g, K)$ is an initial data set which is time-symmetric ($K\equiv 0$),  suppose $\Sigma_0$ is a MOTS in $\mathcal{M}$ with corresponding blow-up solution $f$, and suppose $\Sigma_1$ is another surface lying entirely outside of $\Sigma_0$ with strictly smaller area.
Here ``outside of" implies that $\Sigma_1$ is contained in the component of $\mathcal{M}\setminus\Sigma_0$ having the asymptotically flat end and that $\Sigma_1$ is homologous to $\Sigma_0$.
Let $N = \text{graph}f \subset \mathcal{M} \times \mathbb{R}$, where $\mathcal{M} \times \mathbb{R}$ is equipped with the product metric $g +dt^2$, $t$ a coordinate along the $\mathbb{R}$ factor.

In what follows we employ the following notation: if $S$ is any submanifold of $\mathcal{M} \times \mathbb{R}$ and $I \subset \mathbb{R}$ is any interval, then $S_I$ denotes the portion of $S$ lying in $\mathcal{M} \times I$, i.e.\ $S_I := S \cap (\mathcal{M} \times I)$.
We use $| S |$ to denote the volume of $S$.

Now let $\widetilde{N}(\tau) \subset \mathcal{M} \times \mathbb{R}$ be a 1-parameter family of smooth hypersurfaces satisfying the following:
\begin{enumerate}
  \item $\widetilde{N}(\tau)_{(-\infty, \tau - 1]} \equiv N_{(-\infty, \tau - 1]}$
  \item $\widetilde{N}(\tau)_{[2\tau +1, \infty)} \equiv N_{[2\tau +1, \infty)}$
  \item $\widetilde{N}(\tau)_{[\tau, 2\tau]} \equiv \Sigma_1 \times [\tau, 2\tau]$
  \item $\left| \widetilde{N}(\tau)_{[\tau - 1, \tau]} \right|  + \left| \widetilde{N}(\tau)_{[2\tau, 2\tau + 1]} \right|< C$,  $C$  a constant  independent of $\tau$.
\end{enumerate}
Clearly
\[ |\widetilde{N}(\tau)_{[\tau, 2\tau]}| = |\Sigma_1 |  \tau. \] 
Since $f$ blows up at $\Sigma_0$ and $N$ is asymptotic to the cylinder $\Sigma_0 \times \mathbb{R}$, for large $\tau$, $|N_{[\tau - 1, 2\tau + 1]}|$ approaches $| \Sigma_0 \times [\tau - 1, 2\tau + 1] | = |\Sigma_0|  (\tau + 2)$.  
In particular, for large enough $\tau$ we have
\[ |N_{[\tau - 1, 2\tau + 1]}| > |\Sigma_0|  (\tau + 2) - C^\prime, \]
where $C^\prime  > 0$ is a constant independent of $\tau$.
Then for large $\tau$, we have
\begin{eqnarray*} |\widetilde{N}(\tau)| - |N| & = & |\widetilde{N}(\tau)_{[\tau-1, 2\tau +1]}| - |N_{[\tau-1, 2\tau +1]}| \\
& = & |\Sigma_1 |  \tau + \left| \widetilde{N}(\tau)_{[\tau - 1, \tau]} \right|  + \left| \widetilde{N}_{[2\tau, 2\tau + 1]} \right|
- \left|N_{[\tau-1, 2\tau +1]}\right| \\
& < & |\Sigma_1 |  \tau + C - |\Sigma_0|  (\tau + 2) + C^\prime \\
& = & \left( |\Sigma_1 | - |\Sigma_0|\right)  \tau + C^{\prime\prime},
\end{eqnarray*}
where $C^{\prime\prime}$ is again independent of $\tau$.
Thus since $|\Sigma_1 | - |\Sigma_0| < 0$, for $\tau$ sufficiently large we have $|\widetilde{N}(\tau)| < |N|$, or equivalently,
$|\widetilde{N}(\tau)_{[\tau-1, 2\tau +1]}| < |N_{[\tau-1, 2\tau +1]}|$.

Now, since $K \equiv 0$, the Jang equation $\mathcal{J}[f] =  0$ reduces to the minimal surface equation $\mathcal{H}[f] = 0$.
Define a 2-form on $\mathcal{M} \times \mathbb{R}$ by $\omega := i_n dV$, where $dV$ is the volume form on $\mathcal{M}\times\mathbb{R}$ and $n$ is the downward unit normal to $N = \text{graph}f $, extended to be constant along the $\mathbb{R}$-factor of $\mathcal{M}\times \mathbb{R}$.
Then $\omega$ is a calibration for $N$; that is, $\omega|_N$ is the volume form for $N$, and 
\[ d\omega = d (i_n dV) = (\text{div}\,n) dV \equiv 0,\]
since $\text{div}\, n = \mathcal{H}[f] = 0.$
A calibrated submanifold minimizes volume for its boundary and relative homology class \cite{HL}, so in particular, we should have
\[ |N_{[\tau-1, 2\tau +1]}| \leq |\widetilde{N}(\tau)_{[\tau-1, 2\tau +1]}|\]
for any $\tau$.
But this is a contradiction to the above.
\end{proof}

%%%%%%%%%%%%%%%%%%%%%%%%%%%%%%%%%%%%%%%%%%%%%%%%%%%%%%%%%%%%%%%%%%%%%%%%%%%%%%%%%%%%%%%%%%%%%%%%%%%%%%%%%%%%%%%%%%%%%%%%%%%%%%%%%
%																%
%																%
%			CONSTRUCTION OF A NON-OUTER-AREA-MINIMIZING MOTS W/BLOW-UP SOL'N					%
%																%
%																%
%%%%%%%%%%%%%%%%%%%%%%%%%%%%%%%%%%%%%%%%%%%%%%%%%%%%%%%%%%%%%%%%%%%%%%%%%%%%%%%%%%%%%%%%%%%%%%%%%%%%%%%%%%%%%%%%%%%%%%%%%%%%%%%%%

\begin{prop} There exist spherically symmetric initial data sets with $K$ not identically $0$, satisfying the dominant energy condition, which contain MOTSs which have blow-up solutions but which are not outer-area-minimizing.  
\end{prop}
\begin{proof}  By construction: we find an example of an initial data set $(\mathcal{M} = [\ell_0, \infty), R, K_\ell, K_R)$ containing two MOTSs, say at $\ell_1$, $\ell_2$ where $\ell_0 < \ell_1 < \ell_2$, such that $R(\ell_1) > R(\ell_2)$ but $\ell_1$ has a blow-up solution corresponding to it.
Since $R(\ell)$ is by definition the area radius of the round 2-sphere at $\ell$, $R(\ell_1) > R(\ell_2)$ implies that $\text{Area}(\ell_1) > \text{Area}(\ell_2)$.
Then $\ell_1$ is a non-outer-area-minimizing MOTS with a blow-up solution. 
The main difficulty lies in arranging the data such that the dominant energy condition is satisfied.

For convenience and clarity, we switch to using subscript notation for derivatives in what follows --- e.g.\ $k_{,\ell}$ in place of $\del_\ell k$, et cetera.

We begin by finding a function $a(\ell) \in C^\infty\left([ \ell_0, \infty)\right)$ satisfying the following properties:

\begin{enumerate}[ i. ]
\item $a(\ell) > -\ell$ for all $\ell \in [ \ell_0, \infty)$ 
\item $a_{,\ell}(\ell_1) = a_{,\ell}(\ell_2) = -1 $ 
\item $a_{,\ell}(\ell) < -1 $ for $\ell \in (\ell_1, \ell_2)$ and $a_{,\ell}(\ell) > -1$ for $\ell \in [\ell_0, \infty) \setminus [\ell_1, \ell_2]$ 
\item $a_{,\ell\ell}(\ell_1) = a_{,\ell\ell}(\ell_2) = 0$ 
\item  $|a(\ell)| + \ell \,|a_{,\ell}(\ell)|  + \ell^2 |a_{,\ell\ell}(\ell)| \leq C$ for all $\ell \in [\ell_0, \infty)$, some $C < \infty$
\item  $a_{,\ell}(a_{,\ell} + 2) + 2a_{,\ell\ell}(a+\ell) < 0$ for $\ell \geq L$, some constant $L > \ell_2$.
\end{enumerate}
Property (vi) is readily satisfied if, for example, $a(\ell) = c\ell^{-1}$ for very large $\ell$, where $c$ is any constant.
We now set
\begin{equation} R(\ell) = \ell + a(\ell), \label{R}\end{equation}
from which we derive that
\begin{equation} H(\ell) =  \frac{2(1+a_{,\ell}(\ell))}{\ell + a(\ell)} \label{Ha} \end{equation}
and
\begin{equation} H_{,\ell}(\ell) = \frac{2a_{,\ell\ell}(\ell)}{\ell + a(\ell)} - \textstyle\frac{1}{2}H^2(\ell). \label{dH} \end{equation}
Our assumptions about $a(\ell)$ imply the following:
\begin{enumerate}[   $\circ$   ]
\item  $ R(\ell) > 0  $ for all $\ell$
\item  $ \displaystyle\left| 1 - \frac{R^2(\ell)}{\ell^2} \right| = \left| \frac{2a}{\ell} + \frac{a^2}{\ell^2} \right| = O(\ell^{-1}) $ 
\item  $ \displaystyle\left| H(\ell) - \frac{2}{\ell} \right| = \left| \frac{2\ell a_{,\ell} - 2a}{\ell(\ell+ a)}\right| = O(\ell^{-2})$
\item  $ \left| H_{,\ell}(\ell) + \displaystyle\frac{2}{\ell^2} \right| = \displaystyle\left| \frac{2a_{,\ell\ell}}{\ell + a} + \frac{2a(a +2\ell)}{\ell^2(\ell+a )^2} - \frac{2a_{,\ell} (a_{,\ell}  + 2)}{(\ell+a )^2} \right| = O(\ell^{-3})$.
\end{enumerate}
That is, decay conditions \eqref{decay1}-\eqref{decay3} for the metric $g = d\ell^2 + R^2(\ell) ds^2$ are satisfied.

Next, recalling the constant $L$ from property (vi) of the definition of $a(\ell)$, we choose $K_R(\ell)$ to be a function in $C^\infty\left( [\ell_0, L] \right)$ satisfying the following:

\begin{enumerate}[ i. ]
\item $K_R(\ell_1) = K_R(\ell_2) = 0$  
\item ${K_R}_{,\ell}(\ell_1) = {K_R}_{,\ell}(\ell_2) = 0$  
\item $2K_{R,\ell\ell}(\ell_1) > -H_{,\ell\ell}(\ell_1)$ and $2K_{R,\ell\ell}(\ell_2) > -H_{,\ell\ell}(\ell_2)$  
\item $2K_R(\ell) > |H(\ell)|$ for $\ell \in [\ell_0, L] \setminus \{ \ell_1, \ell_2 \} $.  
\end{enumerate}

Now, the dominant energy condition says that we must have 
\[ \rho \geq |\mu|, \]
where $\rho$ and $\mu_i$ are defined via the Einstein constraint equations, 
\[ 8\pi \rho = \text{Scal} + |K|^2 - (\text{tr}K)^2 \]
and
\[ 8\pi \mu_i = 2 \left(\nabla^j K_{ij} - \nabla_i (\text{tr}K ) \right),\] 
respectively.
Computing these expressions in our spherically symmetric setting, we find that the dominant energy condition amounts to precisely the following inequality:
\begin{equation}
\frac{2}{R^2} - \frac 32 H^2 - 2H_{,\ell} - 4K_\ell K_R - 2K_R^2
\geq
\left| 2H(K_\ell - K_R) - 4K_{R,\ell} \right|. \label{DEC}
\end{equation}
From equations \eqref{Ha} and \eqref{dH}, we see that properties (ii) and (iv) of the definition of $a(\ell)$ imply that $H(\ell_1) = H(\ell_2) =0$ and $H_{,\ell}(\ell_1) = H_{,\ell}(\ell_2) =0$; using in addition properties (i) and (ii) of  $K_R(\ell)$ to evaluate each side of \eqref{DEC} at $\ell_i$, $i = 1, 2$, we have
\[ \text{LHS} = \frac{2}{R^2(\ell_i)} \qquad \text{and} \qquad \text{RHS} = 0.\]
Thus since $R(\ell_i) > 0$ for $i = 1,2$, strict inequality holds in \eqref{DEC} at each $\ell_i$, which in turn implies that the dominant energy condition holds at least on $\mathcal{U}_1 \cup \mathcal{U}_2$, where each $\mathcal{U}_i$ is a neighborhood of $\ell_i$, $i=1,2$.
Next, observe that the inequality 
\begin{equation}  - 4K_\ell K_R - \left| 2HK_\ell \right|
\geq \left| 2H K_R + 4K_{R,\ell} \right|  - \frac{2}{R^2} + \frac 32 H^2 + 2H_{,\ell} + 2K_R^2 \label{DECsuff}
\end{equation}
implies \eqref{DEC}.
Because $[\ell_0, L]$ is compact and all the terms are smooth, the righthand side of \eqref{DECsuff} is uniformly bounded above on $[\ell_0, L]$, say by a constant $C$. 
Let us assume \textit{a priori} that we will choose $K_\ell$ such that $K_\ell(\ell) < 0$ on $[\ell_0, L]$.
Then in order to insure that the dominant energy condition \eqref{DEC} holds on $[\ell_0, L]$,  it is sufficient to show that
\begin{equation}  -2K_\ell ( 2K_R - \left| H  \right| ) \geq C. \label{DECsuff2}
\end{equation}
Since $2K_R -|H|$ is uniformly bounded below by a positive constant on $[\ell_0, L]\setminus(\mathcal{U}_1 \cup \mathcal{U}_2)$ by property (iv) of $K_R$, we may choose a constant $B > 0$ sufficiently large that $K_\ell(\ell) \leq - B$ implies strict inequality in \eqref{DECsuff2} on the set $[\ell_0, L]\setminus(\mathcal{U}_1 \cup \mathcal{U}_2)$.
In fact, for $\ell \in [\ell_0, L]$,  we set 
\[ \text{tr}K(\ell) \equiv -B,\]
 so that 
\[ K_\ell(\ell) := (\text{tr}K - 2K_R)(\ell) = -B - 2K_R(\ell) \leq -B
\]
(since $K_R \geq 0$ by hypothesis).
Thus with these choices of $R$, $K_R$, and $K_\ell$, inequality \eqref{DECsuff2} holds everywhere on $[\ell_0, L]$ with strict inequality, and hence so does the dominant energy condition \eqref{DEC}, also with strict inequality.

It remains to prescribe $K_R$ and $K_\ell$ (or equivalently, $K_R$ and $\text{tr}K$) on $[L, \infty)$.
To this end, we first make the ansatz
\begin{equation}
K_R(\ell) := \frac{1}{R^3(\ell)} \left[ R^3(L) K_R (L) + \int_L^\ell R_{,\ell} R^2  b (\tilde{\ell}) \, d\tilde{\ell} \right]
\end{equation}
for some unknown function $b(\ell) \in C^\infty([L, \infty))$, where the value $K_R(L) > 0$ and the function $R(\ell)$ are defined as in the preceding paragraphs.
We then choose the function $b(\ell)$ to satisfy the following properties:
\begin{enumerate}[ i. ]
  \item $b(L) = - B$
  \item all derivatives of $b$ vanish at $L$, i.e.\ $b_{,\ell}(L) = b_{,\ell\ell}(L) = b_{,\ell\ell\ell}(L) = \cdots = 0$
  \item $b(\ell) < 0$ for all $\ell \geq L$
  \item $|b(\ell)|$ decays sufficiently rapidly to insure that $K_R(\ell) \geq 0$ for all $\ell \geq L$
  \item $|b(\ell)| + \ell |b_{,\ell} (\ell)| = O(\ell^{-3})$.   
\end{enumerate}
Set
\[ \text{tr}K(\ell) := b(\ell),\]
so that
\begin{equation} K_\ell(\ell) := b(\ell) - 2K_R(\ell). \label{Kell}\end{equation}

One readily checks that, with these choices of $K_R$ and $K_\ell$, 
\begin{equation}  2H(K_\ell - K_R) - 4K_{R,\ell} \equiv 0,\label{KRode} \end{equation}
so the righthand side of the dominant energy condition inequality \eqref{DEC} is identically $0$.
Since $K_R \geq 0$ and $\text{tr}K < 0$ for $\ell \geq L$,
\[ - 4K_\ell K_R - 2K_R^2 = -4 (\text{tr}K - 2K_R)K_R - 2K_R^2  = - 4 \text{tr}K + 6 K_R^2  > 0, \]
and by \eqref{R} and property (vi) of the definition of $a(\ell)$, we also have that for $\ell \geq L$, 
\[ \frac{2}{R^2} - \frac 32 H^2 - 2H_{,\ell} = -\frac{2}{(\ell+ a)^2}\left( a_{,\ell}(a_{,\ell} + 2) + 2a_{,\ell\ell}(a+\ell)   \right) > 0. \]
Thus the lefthand side of \eqref{DEC} is strictly positive for all $\ell \geq L$, i.e.\ the dominant energy condition holds on $[L, \infty)$ as well.

Note that since $R_{,\ell}(\ell) = a_{,\ell}(\ell) + 1 > 0$, $b(\ell) < 0$, and $K_R(\ell) \geq 0$ for all $\ell \geq L$, we have
\[ |K_R(\ell)| = K_R(\ell) \leq \frac{R^3(L) K_R (L)}{R^3(\ell)} = O(\ell^{-3}). \]
Since $|\text{tr}K(\ell)| = O(\ell^{-3})$ by property (v) of $b(\ell)$, $|K_\ell(\ell)| = O(\ell^{-3})$ as well.
Then from \eqref{KRode}, property (v) of $b(\ell)$, and the fact that $|H(\ell)| = O(\ell^{-1})$, we see that $|K_{R,\ell}(\ell)|, |K_{\ell,\ell}(\ell)| = O(\ell^{-4})$.
Thus decay conditions \eqref{decay4}-\eqref{decay6} are satisfied.

There remains one complication with this initial data set: as we have prescribed it, it is not quite smooth. 
Clearly $R(\ell)$ is smooth by construction, and property (ii) of $b$ insures that $\text{tr}K$ is everywhere smooth as well, but while the functions $K_R(\ell)$ and hence $K_\ell(\ell)$ are $C^\infty$ on $[\ell_0, \infty)\setminus\{ L \}$, they are only continuous at $L$.
But now we perturb $K_R$ on a small interval $(L - \varepsilon, L)$ to smooth it out at $L$; since strict inequality holds for the dominant energy condition at $L$ by construction, by making our perturbation sufficiently small, we can insure that the dominant energy condition continues to hold.
We can also insure that $K_R$ remains strictly positive near $L$.
Holding $\text{tr}K$ fixed (i.e.\ still identically equal to $-B$), we use \eqref{Kell} to adjust the function $K_\ell$ accordingly.
Now the whole data set $(\mathcal{M} = [\ell_0, \infty), R, K_\ell, K_R)$ is everywhere smooth and satisfies the dominant energy condition.

Finally it remains to show that this data fulfills the statement of the proposition.
First observe that property (iii) of $a(\ell)$  
implies that  $R_{,\ell}(\ell) < 0 $ for $\ell \in (\ell_1, \ell_2)$ and hence that $R(\ell_1) > R(\ell_2)$.
Furthermore, since $H(\ell_1) = H(\ell_2) = K_R(\ell_1) = K_R(\ell_2) = 0$ by construction, \eqref{theta} tells us that $\theta^+(\ell_1) = \theta^+(\ell_2) = 0$ as well, so both $\ell_1$, $\ell_2$ are MOTSs. 
Our construction also implies that $\theta^+$ is strictly positive on $[\ell_0, \infty)\setminus\{ \ell_1, \ell_2 \}$, that
$\theta^+_{,\ell}(\ell_1) =  \theta^+_{,\ell}(\ell_2) = 0$, 
and that $\theta^+_{,\ell\ell}(\ell_1), \theta^+_{,\ell\ell}(\ell_2) > 0$. 
In particular, $\ell_1$ and $\ell_2$ are both finitely stable, $\ell_2$ is globally outermost, and $\ell_1$ is outermost in $[\ell_1, \ell_2)$.
Now since $\text{tr}K < 0$ for all $\ell$, Proposition 2 applies, so there exists a blow-up solution corresponding to $\ell_1$ as asserted.
\end{proof}

\section*{Acknowledgments}

The author wishes to thank Rick Schoen for a number of very helpful conversations. 

\bibliography{mybib}{}
\bibliographystyle{amsplain}

\end{document}